\documentclass{cccg18}
\usepackage{graphicx,amssymb,amsmath}
\usepackage[ruled,vlined]{algorithm2e}

\title{Continuous Terrain Guarding with Two-Sided Guards}

\author{Wei-Yu Lai\thanks{Department of Computer Science and Information Engineering,  National Taiwan University of Science Technology, {\tt D10115005@mail.ntust.edu.tw}}
        \and
        Tien-Ruey Hsiang\thanks{Department of Computer Science and Information Engineering,  National Taiwan University of Science Technology, {\tt  trhsiang@csie.ntust.edu.tw}}}

\index{Lai, Wei-Yu}
\index{Hsiang, Tien-Ruey}

\begin{document}
\thispagestyle{empty}
\maketitle

\begin{abstract}
Herein, we consider the continuous 1.5-dimensional(1.5D) terrain guarding problem with two-sided guarding. We provide an $x$-monotone chain $T$ and determine the minimal number of vertex guards such that all points of $T$ have been two-sided guarded. A point $p$ is two-sided guarded if there exist two vertices $v_i$ (left of $p$) and (right of $p$) that both see $p$. A vertex $v_i$ sees a point $p$ on $T$ if the line segment connecting $v_i$ to $p$ is on or above $T$. We demonstrate that the continuous 1.5D terrain guarding problem can be transformed to the discrete terrain guarding problem with a finite point set $X$ and that if $X$ is two-sided guarded, then $T$ is also two-sided guarded. Through this transformation, we achieve an optimal algorithm that solves the continuous 1.5D terrain guarding problem under two-sided guarding. 
\end{abstract}

\section{Introduction}
A 1.5 dimensional(1.5D)  terrain $T$ is an $x$-monotone polygonal chain in $\mathbb{R}^2$ specified by $n$ vertices $V(T) = \{v_1,..., v_i,..., v_n\}$, where $v_i = (x_i, y_i)$. The vertices induce $n-1$ edges $E(T) = \{ e_1,..., e_i,..., e_{n-1} \}$ with $e_i$ = $\overline{v_iv_{i+1}}$.

A point $p$ sees or guards $q$ if the line segment $\overline{pq}$ lies above or on $T$, or more precisely, does not intersect the open region bounded from above by $T$ and from the left and right by the downward vertical rays emanating from $v_1$ and $v_n$.

There are two types of terrain guarding problems: (1) continuous terrain guarding (CTG) problem, with objective of determining a subset of $T$ with minimum cardinality that guards $T$, and (2) discrete terrain guarding problem, with the objective of determining a subset of $U$ with minimum cardinality guarding $X$, given that the two point sets $U$ and $X$ are on $T$.

Many studies have referred to applications of 1.5D terrain guarding in real world~\cite{EA,finite_set,FA}. The examples include guarding or covering a road with security cameras or lights and using line-of-sight transmission networks for radio broadcasting.

\subsection{Related Work}
Ample research has focused on the 1.5D terrain guarding problem, which can be divided into the general terrain guarding problem and the orthogonal terrain guarding problem.

In a 1.5D terrain, King and Krohn \cite{NP} proved that the general terrain guarding problem is NP-hard through planar 3-SAT.

Initial studies on the 1.5D terrain guarding problem discussed the design of a constant-factor approximation algorithm. Ben-Moshe et al. \cite{F-app} gave the first constant-factor approximation algorithm for the terrain guarding problem and left the complexity of the problem open. King~\cite{5-app} gave a simple 4-approximation, which was later determined to actually be a 5-approximation. Recently, Elbassioni et al. \cite{4-app} gave a 4-approximation algorithm.

Finally, Gibson et al. \cite{D-PTAS} considered the discrete terrain guarding problem by finding the minimal cardinality from candidate points that can see a target point \cite{D-PTAS} and proved the presence of a planar graph that appropriately relating the local and global optima; thus, the discrete terrain guarding problem allows a polynomial time approximation scheme (PTAS) based on local search. Friedrichs et al. \cite{C-PTAS} revealed that for the continuous 1.5D terrain guarding problem, finite guard and witness sets ($G$ and $X$, respectively) can be constructed such that an optimal guard cover $G'' \subseteq G$ that covers terrain $T$ is present and when these guards monitor all points in $X$, the entire terrain is guarded. According to \cite{D-PTAS}, the continuous 1.5D terrain guarding problem can apply PTAS by constructing a finite guard and witness set with the former PTAS.

Some studies have considered orthogonal terrain $T$. $T$ is called an orthogonal terrain if each edge $e\in E(T)$ is either horizontal or vertical. An orthogonal terrain has four vertex types. If $v_i$ is a vertex of orthogonal terrain and the angle $\angle{v_{i-1}v_iv_{i+1}}=\pi/2$, then $v_i$ is a convex vertex, otherwise it is a reflex vertex. A convex vertex $v_i$ is left(right) convex if $\overline{v_{i-1}v_i}$($\overline{v_iv_{i+1}}$) is vertical. A reflex vertex $v_i$ is left(right) reflex if $\overline{v_{i-1}v_i}$($\overline{v_iv_{i+1}}$) is horizontal.

Katz and Roisman \cite{2-OTG} gave a 2-approximation algorithm for the problem of guarding the vertices of an orthogonal terrain. The authors constructed a chordal graph demonstrating the relationship of visibility between vertices. On the basis of \cite{F.Gavril}, \cite{2-OTG} gave a 2-approximation algorithm and used the minimum clique cover of a chordal graph to solve the right(left) convex vertex guarding problem. 

Lyu and {\" U}ng{\"o}r~\cite{nlogm} gave a 2-approximation algorithm for the orthogonal terrain guarding problem that runs in $O(n\log m)$, where $m$ is the output size. The authors also gave an optimal algorithm for the subproblem of the orthogonal terrain guarding problem. On the basis of the vertex type of the orthogonal terrain, the objective of the subproblem is to determine a minimum cardinality subset of $V(T)$ guarding all right(left) convex vertices of $V(T)$; furthermore, the optimal algorithm uses stack operations to reduce time complexity.

The $O(n\log m)$ time 2-approximation algorithm has previously been considered the optimal algorithm for the orthogonal terrain guarding problem. However, some studies have used alternatives to the approximation algorithm.

Durocher et al. \cite{O-OTG} gave a linear-time algorithm for guarding the vertices of an orthogonal terrain under a directed visibility model, where a directed visibility mode considers the different visibility for types of vertex. If $u$ is a reflex vertex, then $u$ sees a vertex $v$ of $T$, if and only if every point in the interior of the line segment $uv$ lies strictly above $T$. If $u$ is a convex vertex, then $u$ sees a vertex $v$ of $T$, if and only if $\overline{uv}$ is a nonhorizontal line segment that lies on or above $T$. Khodakarami et al. \cite{range} considered the guard with guard range. They presented a fixed-parameter algorithm that found the minimum guarding set in time $O(4^k\cdot k^2 \cdot n)$, where $k$ is the terrain guard range.

\subsection{Result and Problem Definition}
In this paper, we define the CTG problem with two-sided guards and propose an optimal algorithm for the 1.5D CTG problem with two-sided guards. To the best of our knowledge, the 1.5D CTG problem with two-sided guards has never been examined.

\begin{figure}
\begin{center}
\includegraphics[scale=0.6]{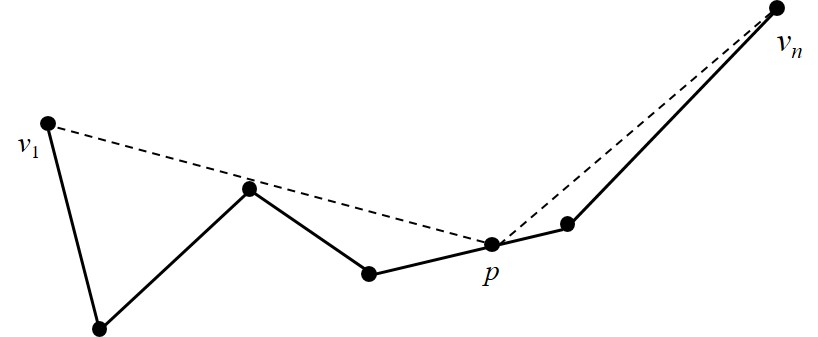}
\caption{Point $p$ is two-sided guarded by $v_1$ and $v_n$}
\label{fig1}
\end{center}
\end{figure}

{\bf Definition 1} (Two-Sided Guarding). A point $p$ on a 1.5D terrain is two-sided guarded if there exist two distinct guards $u$, which is on or to 
the left of $p$, and $v$, which is on or to the right of $p$, such that $p$ can be seen by both $u$ and $v$. Furthermore, 
the guards $u$ and $v$ are called a left guard and a right guard of $p$.

Fig.~\ref{fig1} illustrates an example where vertex $v_1$ left guards $p$ and $v_n$ right guards $p$. In this paper, we define the following problem:

{\bf Definition 2}(CTGTG: Continuous Terrain Guarding with Two-Sided Guards) Given a 1.5D terrain $T$, find a vertex guard set $S$ of minimum cardinality such that every point of $T$ can be two-sided guarded. 

\subsection{Paper Organization}
Section 2 presents preliminaries, Section 3 demonstrates how to create a finite point set for the CTGTG model, Section 4 gives an algorithm for the CTGTG, along with its proof, and Section 5 presents our conclusions.

\section{Preliminaries}
Let $p$ and $q$ be two points on a 1.5D terrain, 
we write $p\prec q$ if $p$ is on the left of $q$. We denote the visible region of $p$ by 
$vis(p) = \{v|v\in V(T)$ and $v$ sees $p\}$. For a $vis(p)$, let $L(p)$ be the leftmost vertex in $vis(p)$ and $R(p)$ be the rightmost vertex in $vis(p)$.

Given a CTGTG instance, let $OPT=\{o_1,o_2,...,o_m\}$ be an optimal guard set, where $o_k\prec o_{k+1}$ for $k=1,...,m-1$.
For a point $p$ on the terrain, let $O_R(p)$ and $O_L(p)$ be the subsets of $OPT$ such that $p$ is right guarded by every guard in $O_R(p)$ and left guarded by every guard in $O_L(p)$.
We also define $N_i^R$ as the rightmost point on the terrain that is not right guarded by 
$\{ o_i,o_{i+1},...,o_m\}$ and $N_i^L$ as the leftmost point on the terrain that is not left guarded by $\{ o_1,o_2,...,o_i\}$.



An important visible property on 1.5D terrains is as follows:

\begin{lemma}[\cite{F-app}]
\label{order lemma}
Let $a$, $b$, $c$ and $d$ be four points on a terrain $T$ such that $a \prec b \prec c \prec d$. If $a$ sees $c$ and $b$ sees $d$, then $a$ sees $d$.
\end{lemma}

\begin{figure}
\begin{center}
\includegraphics[scale=0.7]{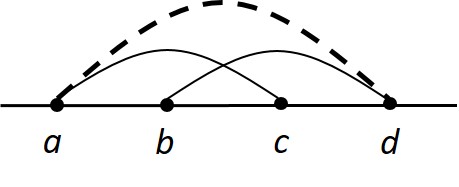}
\caption{Schematic of Lemma~\ref{order lemma}.}
\label{fig2}
\end{center}
\end{figure}

Fig.~\ref{fig2} is a schematic of Lemma 1. Because $T$ is an $x$-monotone chain, we use a straight line to demonstrate the relation between $x$-coordinate of points and an arc to show the visible relation among points on $T$. In this report, we use a straight line to simplify the explanations.


\begin{obs}
\label{obs_l}
Assume point $x$ is on $e_j$. If $x$ is left guarded by $v$ then $v_{j+1}$ is also left guarded by $v$. 
\end{obs}

\begin{obs}
\label{obs_r}
Assume point $x$ is on $e_j$. If $x$ is right guarded by $v$ then $v_{j}$ is also right guarded by $v$. 
\end{obs}

\section{Discretization}

\begin{figure}
\begin{center}
\includegraphics[scale=0.7]{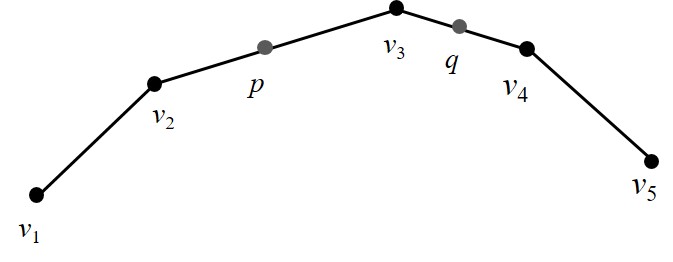}
\caption{$V(T)$ is right guarded and left guarded by $\{v_1, v_2, v_4, v_5\}$, but not $T$.}
\label{V(T) is not enough}
\end{center}
\end{figure}

Although $V(T)$ are right guarded and left guarded, $T$ is not necessarily right guarded and left guarded. In Fig. \ref{V(T) is not enough}, $V(T)$ is right guarded and left guarded by $\{v_1, v_2, v_4, v_5\}$ with minimal cardinality. The vertices $v_1$ and $v_2$ are left guarded by $v_1$ and right guarded by $v_2$. Vertices $v_4$ and $v_5$ are left guarded by $v_4$ and right guarded by $v_5$. Vertex $v_3$ is left guarded and right guarded by $v_2$ and $v_4$, respectively. Only $v_3$ can right guard $p$ and left guard $q$ where $p$ is on $e_2$ and $q$ is on $e_3$, but $v_3 \notin \{v_1, v_2, v_4, v_5 \}$. In our example, we must create a point set $X$ such that if $X$ is right guarded and left guarded, then $T$ is also.

{\bf Definition 3} (Boundary Point). If line $\overline{v_iv_j}$ and $e_k$ have an intersection point $f\notin \{ v_k, v_{k+1} \}$, and $v_i$ and $v_j$ can see $f$ then $f$ is the boundary point.

In Fig.~\ref{fig:7}, we provide an example with four boundary points: $f_1, f_2, f_3$ and $f_4$. Boundary point $f_1$ is from $v_6$, $f_2$ is from $v_4$; and boundary points $f_3$ and $f_4$ are from $v_1$. We say $e_1$ has two boundary points, $f_1$ and $f_2$; each of $e_4$ and $e_6$ has a boundary point. 

\begin{figure}
\begin{center}
\includegraphics[scale=0.5]{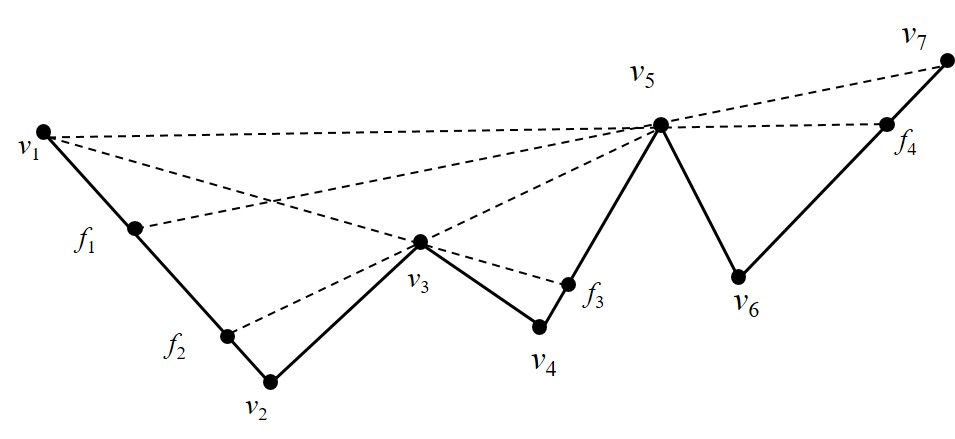}
\caption{Points $f_1$, $f_2$, $f_3$ and $f_4$ are boundary points on $T$.}
\label{fig:7}
\end{center}
\end{figure}

\begin{lemma}
\label{x lemma}
For an edge $e_i$ on terrain $T$, there exists at most two points $p$ and $q$ that exclude $v_i$ and $v_{i+1}$ such that $e_i$ is complete two sided guarded if $p$ and $q$ are two sided guarded.
\end{lemma}
\begin{proof}

According to the number of boundary points on $e_i$, we may consider the proof under the following heads: edge $e_i$ does not have boundary point or has one, two, or $k$ boundary points (where $k \geq$ 3).

In the first case, we assume $e_i$ does not have boundary point. Let point $p \notin \{ v_i,v_{i+1} \}$ be on edge $e_i$. If $p$ is right guarded and left guarded, then edge $e_j$ is also right guarded and left guarded. 

In the second case, we assume $e_i$ has a boundary point $f$. We divided the edge into two line segments $\overline{v_i f}$ and $\overline{v_{i+1} f}$. Then, the first case can be applied to the line segments $\overline{v_i f}$ and $\overline{f v_{i+1}}$. Therefore, we create two points $p \notin \{ v_i,f,v_{i+1}\}$ on line segment $\overline{v_i f}$ and $q \notin \{ v_i,f,v_{i+1}\}$ on line segment $\overline{fv_{i+1}}$. If $p$ and $q$ are right guarded and left guarded, then $e_i$ is also right guarded and left guarded. 

In the third case, we assume $e_i$ has two boundary points $f_1$ and $f_2$. We divided the edge into three line segments $\overline{v_if_1}$, $\overline{f_1f_2}$ and $\overline{f_2v_{i+1}}$. The line segments $\overline{v_if_1}$ and $\overline{f_2v_{i+1}}$ can be reduced to the first case. Therefore, we create two points $p \notin \{ v_i,f_1 \}$ on line segment $\overline{f_2v_{i+1}}$ and  $q \notin \{  f_2,v_{i+1} \}$ on line segment $\overline{f_2v_{i+1}}$.
If $p$ and $q$ are left guarded and right guarded, then line segment $\overline{f_1f_2}$ is also left guarded and right guarded.

In the final case, we assume $e_i$ has $k$ boundary points $f_1,...,f_k$. We divide the edge into $k+1$ line segments $L = \{ \overline{v_if_1},\overline{f_1f_2},...,\overline{f_k v_{i+1}} \} $. The line segments $\overline{v_if_1}$ and $\overline{f_kv_{i+1}}$ can be reduced to the first case. 
Therefore, we create two points: $p \notin \{ v_i, f_1 \}$ on line segment $\overline{v_if_1}$ and $q \notin \{ f_k, v_{i+1} \}$ on line segment $\overline{f_kv_{i+1}}$.  
If $p$ and $q$ are left guarded and right guarded, then each line segment $\overline{f_cf_{c+1}} \in L$ is also left and right guarded.
\end{proof}

From the construction of Lemma~\ref{x lemma}, in order to completely two-sided guard a terrain, it is sufficient to first select a finite subset $X$ of positions from the terrain to be two-sided guarded, such that $|X|\leq 2(n-1)$.

\section{4. Optimal Algorithm for the CTGTG}

In this section, we present an optimal algorithm for the CTGTG. The idea of the algorithm follows from Observation~\ref{include $v_1$ and $v_n$}. In each step of our algorithm, we add a vertex $v_i$ to our result $S$ such that if $v_i \notin OPT$ then $v_i$ can replace a vertex $v_j \in OPT$ and $|S| = |OPT|$.     
 
\begin{obs}
\label{include $v_1$ and $v_n$}
The optimal solution of the CTGTG includes $v_1$ and $v_n$. 
\end{obs}

This is because in the CTGTG for right and left guarded $T$, only $v_1$ can left guard $v_1$ and only $v_n$ can right guard $v_n$.

\begin{figure}
\begin{center}\includegraphics[scale=0.8]{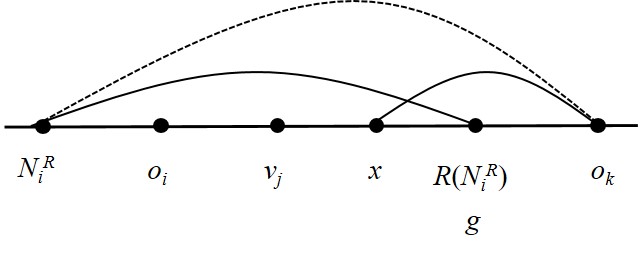}
\caption{Position of $R(N^R_i)$ and $g \in O_R(N^R_i)$.}
\label{fig3}
\end{center}
\end{figure}

\begin{lemma}
\label{between}
$R(N^R_i)$ and $ g \in O_R(N^R_i)$ do not lie on the right side of $o_i$.
\end{lemma}

\begin{proof}
Assume $R(N^R_i)$($g \in O_R(N^R_i)$) is on the right side of $o_i$ and $x$ is on the edge $e_j =\overline{v_jR(N^R_i)}$($\overline{v_jg}$). We know that $x$ is right guarded by $o_k$ and $o_k$ is on the right side of $R(N^R_i)(g)$. According to Lemma \ref{order lemma}, if $o_k$ right guards $x$, then $N^R_i$ is right guarded by $o_k$. This contradicts the definition of $N^R_i$ and $o_k$ sees $N^R_i$. The schematic of Lemma~\ref{between} is given in Fig~\ref{fig3}.
\end{proof}

\begin{lemma}
\label{between_l}
$L(N^L_i)$ and $g \in O_L(N^L_i)$ do not lie the left side of $o_i$.
\end{lemma}

\begin{lemma}
\label{can't guard}
If $R(N^R_i) \notin O_R(N^R_i)$ and 
$R(N^R_i) \prec x$,
then $g \in O_R(N^R_i)$ cannot left guard $x$.
\end{lemma}

\begin{proof}
We prove Lemma~\ref{can't guard} in two steps. The first step explains that if $g \in O_R(N^R_i)$ left guards $x_k \in \{x_j \mid R(N^R_i) \prec x_j \}$, then $x_k$ and $N^R_i$ can see each other. The second step explains that if $x_k$ and $N^R_i$ can see each other, and $g$ left guards $x_k$, then $N^R_i$ is right guarded by $g$.
First, let $N^R_i \prec g \prec R(N^R_i) \prec x_k$. Because $R(N^R_i)$ sees $N^R_i$, according to Lemma~\ref{order lemma} if $x$ and $g$ see each other, then $x_k$ and $N^R_i$ also see each other. This is illustrated in Fig.~\ref{fig:4}.

\begin{figure}
\begin{center}
\includegraphics[scale=0.8]{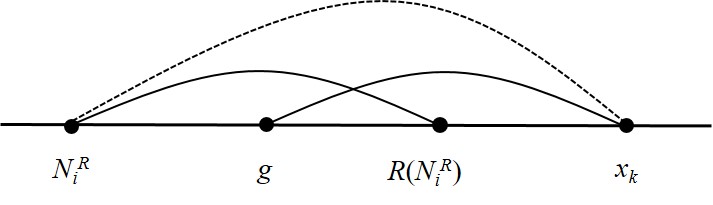}
\caption{If $O_R(N^R_i)$ left guard $x_k$, then $x_k$ and $N^R_i$ see each other.}
\label{fig:4}
\end{center}
\end{figure}

In the second step, we assume that $x_k$ is right guarded by $o_j$ and that $x_k$ is on the edge $e_k$ on the right side of $R(N^R_i)$. We know that if $O_R(N^R_i)$ left guards $x_k$, then $x_k$ and $N^R_i$ see each other. Because $o_j$ right guards $x_k$ and sees $v_l$, if $x$ sees $N^R_i$ then $o_j$ right guard $N^R_i$ too, as illustrated in Fig.~\ref{fig:5}.

\begin{figure}
\begin{center}
\includegraphics[scale=0.8]{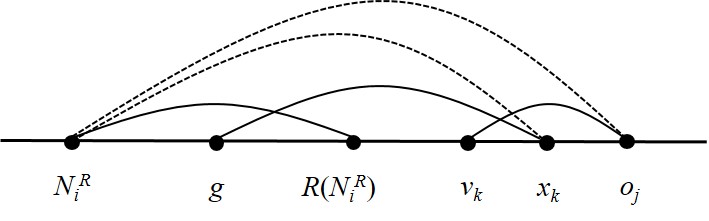}
\caption{If $x_k$ and $N^R_i$ see each other, then $o_j$ right guards $N^R_i$.}
\label{fig:5}
\end{center}
\end{figure}

\end{proof}

\begin{lemma}
\label{can't guard l}
If $L(N^L_i) \notin O_L(N^L_i)$ and $x \prec L(N^L_i)$, then $g \in O_L(N^L_i))$ cannot right guard $x$.
\end{lemma}

\begin{lemma}
\label{can't guard_1}
If $R(N^R_i) \notin O_R(N^R_i)$, $x$ is right guarded by $o_j$ and $i \leq j \leq m$, then $x$ cannot lie between $g \in O_R(N^R_i)$ and $R(N^R_i)$.
\end{lemma}

\begin{proof}
We assume $x$ is on the $e_k = \overline{v_kR(N^R_i)}$ and $x$ is right guarded by $o_j$. We know that $o_j$ right guards $v_k$ by Observation 2. According to Lemma~\ref{order lemma}, if $x$ is right guarded by $o_j$ then $N_R(o_i)$ is right guarded by $o_j$. Therefore, we know that if $R(N^R_i) \notin O_R(N^R_i)$, then $x$ cannot lie between $g \in O_R(N^R_i)$ and $R(N_R(o_i))$.

\end{proof}

\begin{lemma}
\label{can't guard l_1}
If $L(N_L(o_i)) \notin O_L(N_L(o_i))$, $x$ is left guarded by $o_j$ and and $1 \leq j \leq i$, then $x$ cannot lie between $g \in O_R(N^R_i)$ and $R(N^R_i)$.
\end{lemma}

\begin{theorem}
If $R(N^R_i) \notin O_R(N^R_i)$, then $R(N^R_i)$ can replace $g \in O_R(N^R_i)$.
\end{theorem}
\begin{proof}
Based on Lemma~\ref{between}, Lemma~\ref{can't guard} and Lemma~\ref{can't guard_1},
if $g \in O_R(N^R_i)$ and $R(N^R_i) \notin O_R(N^R_i)$, then $g$ cannot left guard $x_k \in \{ x_j  \mid N^R_i \prec x_j \} $. Due to $g \prec R(N^R_i))$, we know $vis(R(N^R_i)) \supseteq vis(g)$ by Lemma~\ref{order lemma}.  
\end{proof}

\begin{theorem}
If $O_L(N^L_i) \notin L(N^L_i)$, then $L(N^L_i)$ can replace $g \in O_L(N^L_i)$.
\end{theorem}


\section{Complexity}
Because our approach has two phases, we must first discuss the complexity of discretization. We obtain boundary points for a vertex $v$ on $E(T)$ in $O(n)$ by \cite{map}. Therefore, we compute all boundary points for each vertex of $V(T)$ on each edge $e \in E(T)$ in $O(n^2)$. We obtain at most 2$|V(T)|$ boundary points in $O(n^2)$.

\begin{figure}
\begin{center}
\includegraphics[scale=0.8]{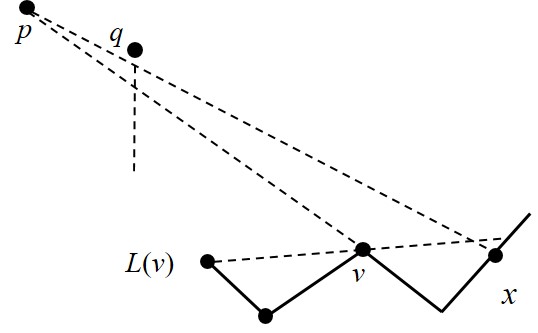}
\caption{If $L(v)$ cannot see $x$ and $v$ sees $x$ then $v=L(x)$.}
\label{fig:8}
\end{center}
\end{figure}

Next, we demonstrate how to compute an optimal solution for the CTGTG. In step 1, we add $v_1$ and $v_n$ to our solution. In step 2, we compute the $vis(v_1)$ and $vis(v_n)$. In step 3, we add $R(x)$ to our solution, where $x$ is the nonright-guarded rightmost point. If a point $x$ exists that is not right guarded, then repeat step 3 until $X$ is right guarded. In step 4, we add $L(x)$ to our solution, where $x$ is the non-left guarded leftmost point. If a point $x$ exists that is not left guarded, then repeat step 4 until $X$ is right guarded. Finally, all points $x$ are right guarded and left guarded.



We show our algorithm for the CTG problem runs in $O(n)$ using two steps. Before the algorithm begins, we can compute $R(x)$ and $L(x)$ for each point of $X$ in $O(n)$. After this computation, we proceed to the algorithm in $O(n)$. Therefore, our proposed algorithm for the CTG problem runs in $O(n)$.


\begin{algorithm}[h]

  \SetAlgoNoEnd
  \caption{Compute all $L(x)$}
  \label{Compute all $L(x_i)$}
  \KwIn{$T$: terrain, $X$: point set}
  \KwOut{\{ $L(x_i)|x_i \in X$ \} }
  
  $Q \leftarrow X \cup V(T)$
  
  \For{$q_i \in Q$ processed from left to right } {
    $q_j = q_{i-1}$
    
    \While{$L(q_i) = \emptyset$}{
    \If{$q_i$ sees $L(q_j)$}{
    \If{$L(q_j)$ is not $v_1$}{
    $q_j=L(q_j)$
    }
    \Else{
    $L(q_i)=v_1$
    }
    }
    \Else{
    $L(q_i)=q_j$
    }
    }
  }
  \For{$x_i \in X$ processed from left to right } {
  Return $L(x_i)$
  }
\end{algorithm}

\begin{lemma}
\label{L(xi)}
If $L(v)$ cannot see $x$ and $v$ sees $x$ then $v=L(x)$. 
\end{lemma}
\begin{proof}
Assume $p \prec L(v)\prec v \prec x$, $L(v)$ cannot see $x$ and $v$ can see $x$. If $p$ sees $x$ and cannot see $v$, then a vertex $q$ exists and lie above line $\overline{vL(v)}$ and $p \prec q \prec L(v)$, as illustrated in Fig.~\ref{fig:8}. However, the assumption that $L(v) \neq q$ is contradictory.   
\end{proof}



We propose Algorithm~\ref{Compute all $L(x_i)$} to compute $L(x)$ for all points of $X$ in $O(n)$ according to Lemma~\ref{L(xi)} and Lemma~\ref{order lemma}. We unite $X$ and $V(T)$ in a set $Q$. Algorithm~\ref{Compute all $L(x_i)$} finds $L(q_i) \in Q$ from left to right.
We prove that the running time of Algorithm~\ref{Compute all $L(x_i)$} is $O(n)$.

\begin{theorem}
Algorithm~\ref{Compute all $L(x_i)$} runs in O$(n)$.
\end{theorem}
\begin{proof}
We count the number of times $q_i$ sees $L(q_j)$ in the algorithm.
If $q_i$ sees $L(q_j)$, then the algorithm does not visit the vetrices between $q_i$ and $L(q_j)$. Therefore, the number of times $q_i$ sees $L(q_j)$ is at most once for each point of $Q$. If $q_i$ does not see $L(q_j)$, then $q_i$ has found $L(q_i)$. Therefore, the number of times $q_i$ does not see $L(q_j)$ is at most once for each point $Q$.
\end{proof}



After computing $L(x_i)$ and $R(x_i)$ for $X$, we reach the algorithm for the CTGTG in $O(n)$. We divided our algorithm into left and right-guarding, and therefore we provide the algorithm for left-guarding that can be implemented in $O(n)$.

\begin{algorithm}[h]
  \SetAlgoNoEnd
  \caption{Left-guarding}
  \label{A2}
  \KwIn{$T$: terrain}
  \KwIn{$X$: point set}
  \KwOut{$P$ : left-guarding set}
  
  $P$ is null;
  
  $V(T')$=$V(T)$;
  
  \For{$x_i \in X$ processed from left to right } {
    \While {$g(x_i)$ is null}{
    $p_j$ is rightmost vertex in $P \cap V(T')$;
    
     \If{$x_i$ is guarded by $p_j$}
     {
      $g(x_i)$ is $p_j$;
      
      Remove the vertices between $x_i$ and $p_j$ from $V(T')$;
     }
     \ElseIf{$p_j$ on the left side of $L(x_i)$}
       {
         $g(x_i)$ be the vertex $L(x_i)$ ;
         
         Add $g(v_i)$ to $P$;
         
         Remove the vertics between $x_i$ and $L(x_i)$ from $V(T')$;
       }
       \Else
       {Remove $p_j$ from $V(T')$;}
    }     	  
    	}
  \Return $P$
\end{algorithm}

\begin{theorem}
Algorithm~\ref{A2} runs in $O(n)$.
\end{theorem}
\begin{proof}
For each $x_i$, we examine whether $x_i$ is guarded by $p_a \in P$ from $x_i$ to $g(x_i)$. If $g(x_i)$ = $g_a$ = $v_b$, then Algorithm~\ref{A2} will not visit the point and vertex between $x_i$ and $v_b$. 
We count the number of times $x_i$ is not seen by $P$.
We can check $p_j$ from $x_i$ to $L(x_i)$. 
If $p_j$ does not see $x_i$, then we will not check $p_j$ for $\{ x_k \mid x_i \prec x_k \}$. 
Assume $p_j$ does not see $x_i$, $p_k$ sees $x_i$, and $p_k \prec p_j \prec x_i$, if $\{x_l \mid x_i \prec x_l \}$ is seen by $p_j$, then $p_k$ sees $\{x_l \mid x_i \prec x_l\}$ according to Lemma~\ref{order lemma}.
The number of times $X$ is not seen by $P$ is $|V(T)|$, and the number of times $X$ is seen by $P$ is $|X|$.
Therefore, the algorithm visits the point and vertex at most $2|X|+|V(T)|$ times. After computing all $L(x_i)$, Algorithm~\ref{A2} runs in $O(n)$.
\end{proof}  

\section{Conclusion}
In this paper, we considered the CTGTG problem and devised an algorithm that can determine the minimal cardinality vertex that guards $T$ under two-sided guarding. We showed that the CTGTG problem can be reduced to the discrete terrain guarding problem with at most $2|V(T)|$ points in $O(n^2)$ and solved the problem using our devised algorithm in $O(n)$ where $n$ is the number of vertices on $T$.

\bibliographystyle{elsarticle-num.bst}
\bibliography{bibliography.bib}

\end{document}